\documentclass[journal]{IEEEtran}  

\usepackage[all]{xy}
\usepackage{graphics,epstopdf,wrapfig}
\usepackage{epsfig}
\usepackage{amsmath}
\usepackage{mathrsfs}
\usepackage{tikz}
\usetikzlibrary{automata,positioning,shapes,arrows}
\usepackage{amssymb,latexsym,amsfonts,amsmath}
\usepackage{graphicx}
\usepackage{subfigure}
\usepackage{color}
\usepackage{epstopdf}
\usepackage{comment}
\usepackage{float}
\usepackage[lined,ruled,commentsnumbered]{algorithm2e}

\usepackage{amsthm}

\newtheorem{thm}{Theorem}

\newtheorem{assumption}{Assumption}
\theoremstyle{remark}

\makeatletter
\newcommand{\removelatexerror}{\let\@latex@error\@gobble}
\makeatother
\newenvironment{definition}[1][Definition]{\begin{trivlist}
		\item[\hskip \labelsep {\bfseries #1}]}{\end{trivlist}}
\newenvironment{problem}[1][Problem]{\begin{trivlist}
		\item[\hskip \labelsep {\bfseries #1}]}{\end{trivlist}}

\usetikzlibrary{calc}
\usepackage{relsize}

\tikzset{fontscale/.style = {font=\relsize{#1}}
}
\title{\LARGE \bf Permissive Supervisor Synthesis for Markov Decision Processes through Learning}
\author{Bo~Wu, Xiaobin~Zhang and Hai~Lin
	\thanks{This work is supported by NSF-CNS-1239222, NSF-EECS-1253488 and NSF-CNS-1446288}
	\thanks{ Bo Wu, Xiaobin Zhang and Hai Lin are with the Department of Electrical Engineering, University of Notre Dame, Notre Dame,
		IN, 46556 USA. {\tt\small bwu3@nd.edu, xzhang11@nd.edu, hlin1@nd.edu}}}
\begin{document}
	\maketitle
	\begin{abstract}
This paper considers the permissive supervisor synthesis for probabilistic systems modeled as Markov Decision Processes (MDP). Such systems are prevalent in power grids, transportation networks, communication networks  and robotics. Unlike centralized planning and optimization based planning, we propose a novel supervisor synthesis framework based on learning and compositional model checking to generate permissive local supervisors in a distributed manner. With the recent advance in assume-guarantee reasoning verification for probabilistic systems, building the composed system can be avoided to alleviate the state space explosion and our framework learn the supervisors iteratively based on the counterexamples from verification. Our approach is guaranteed to terminate in finite steps and to be correct.
	\end{abstract}
\begin{IEEEkeywords}
	formal methods, supervisor synthesis, model checking, automata learning
\end{IEEEkeywords}
	
	\section{Introduction}
	\indent Control and planning of probabilistic systems modeled as Markov Decision Processes (MDP) \cite{rutten2004mathematical} has gained tremendous popularity in recent years because it not only considers the uncertainties in actuation and environments for each individual agent but also the coordination among them for achieving a global specification. The development of a group of several specialized agents could offer better efficiency, fault-tolerance, and flexibility due
	to its distributed operations, parallelized executions and reconfigurability in face of changes.\\
	\indent Since multiple MDPs can be composed into one, a centralized planning approach is the probabilistic model checking \cite{baier2008principles} to obtain the global optimal strategy (policy) for single MDP such as in \cite{kuvcera2005controller,baier2004controller} with specifications in probabilistic linear temporal logic (PLTL) \cite{lahijanian2010motion} and computation tree logic (PCTL) \cite{hansson1994logic}. However, a global controller supervising all the agents may not be feasible. Another planning approach is based on the Decentralized MDP (Dec-MDPs) \cite{bernstein2002complexity} where each agent make decisions with its local information. Most existing results focus on seeking the (approximately) optimal action policies through reinforcement learning or Q-learning \cite{oliehoek2008optimal}, or alternatively converting the problem into stochastic games (SG) \cite{filar2012competitive} but such solutions can be computationally expensive and often intractable if not undecidable.\\
	\indent This paper aims to develop a formal synthesis framework with MDPs so that a certain desired performance can be guaranteed with the designed local supervisors. Without loss of optimality in the bounded time properties that we consider \cite{forejt2011automated}, the supervisors are in the form of the deterministic finite automata (DFA) \cite{cassandras2008introduction} that regulate each agent's action. Unlike the optimization based approach which searches for a single optimized control policy, we are looking for \emph{permissive} local supervisors that potentially enable multiple choices in each step so that the system is resilient against unexpected changes in the runtime or additional constraints may be imposed on the system. Similar ideas was studied in \cite{ujma2015permissive} where memoryless (making decision based only on current state) permissive strategies were generated incrementally using mixed integer linear programming (MILP) \cite{schrijver1998theory} for the single agent. The agent and its environment were modeled as turn based stochastic two-player games. In contrast, our approach considers strategies with memory and works in a decremental fashion by iteratively eliminating counterexample strategies from model checking and refining the local supervisor with L* learning algorithm \cite{angluin1987learning}.\\
	\indent In recent years, counterexample guided approaches have been extensively studied in the aspect of abstraction refinement \cite{hermanns2008probabilistic,clarke2000counterexample} and synthesis \cite{lin2011counterexample,dai2014automatic,raman2015reactive}. In this paper, we mainly focus on the latter. The main purpose of the counterexample guided synthesis is to iteratively learn a correct supervisor under whose regulation the system is guaranteed to satisfy a given specification. In each intermediate iteration,  a guessed supervisor is automatically generated whose correctness will be verified by the model checker and other oracles and refined based on the counterexamples from verification. Such verification-refinement loop ends when no counterexample is generated and thus the guessed supervisor is verified to be correct. \\
	\indent Synthesis of multiple agents is not straightforward for systems modeled as MDPs. Unlike the centralized probabilistic model checking approach, we propose to apply the existing assume guarantee reasoning (AGR) framework for the probabilistic systems \cite{hermanns2008probabilistic,komuravelli2012assume,he2016learning} to alleviate the computation burden. While our framework does not require a particular assume-guarantee reasoning algorithm, we can choose any sound and complete assume-guarantee reasoning algorithms and potentially benefit from state space reduction during the model checking process. The second challenge is to process the counterexamples returned from the compositional model checking. To this end, we propose to label the actions such that we can identify which individual agent causes the specification violation. L* learning is adopted to each single agent based on the counterexample it receives to iteratively refine the corresponding local supervisor. The correctness and termination of our design procedure are proved.  \\
	\indent This paper merges and further develops our previous results \cite{wu2015counterexample,wu2016counterexample}. While the same  L* learning algorithm \cite{angluin1987learning} was considered in synthesis process, we changed the form of the supervisor and considered positive counterexamples so that the results can become less conservative. We also changed the counterexample selection scheme to potentially lower the complexity. Furthermore, we developed the supervisor synthesis software in C++ using counterexample generation tool COMICS \cite{jansen2012comics} and L* learning library libalf \cite{bollig2010libalf}.\\
	\indent This paper is divided into five parts. We give the necessary preliminaries in Section \ref{sec:Background} followed by problem formulation in Section \ref{sec:problem_formulation}. We describe our proposed supervisor synthesis framework in Section \ref{sec:CGSSF}. Section \ref{sec:conclusion} concludes the paper.

\section{Background}
\label{sec:Background}
This section introduces the required background knowledge about MDP, model checking \cite{baier2008principles}, and the L* algorithm. 
\subsection{Markov Decision Process}
\begin{definition}
	An MDP is a tuple $\mathcal{M}=(S,\hat{s},A,T,L)$ where
	\begin{itemize}
		\item $S=\{s_0,s_1,...\}$ is a finite set of states;
		\item $\hat{s}\in S$ is the initial state;
		\item $A$ is a finite set of actions;
		\item $T(s,a,s'):=Pr(s'|s,a),~\forall i\geq0$.
		\item $L:S\rightarrow2^{AP}$ is the labeling function that maps each $s\in S$ to one or several elements of a set $AP$ of atomic propositions.
	\end{itemize}
\end{definition}
%
%
%
%
%
%
%
%
%
%
%
%
%
\indent For each state $s\in S$, we denote $A(s)$ as the set of available actions. From the definition it is not hard to see that the Discrete Time Markov Chain (DTMC) is a special case of MDP with $|A(s)|=1$ for all $s\in S$, where $|A(s)|$ is the cardinality of the set $A(s)$. \\
\indent A path $\omega$ of an MDP is a non-empty sequence of the form $\omega=s_0\xrightarrow{a_0}s_1\xrightarrow{a_1}s_2...s_i\xrightarrow{a_i}s_{i+1}...$, where each transition is enabled by an action $a_i$ such that $T(s_i,a_i,s_{i+1})>0$. We denote $Path^{fin}_s$ as the collection of finite length paths that start in a state $s$. We also denote $\omega^a$ as the action path that's embedded in $\omega$.  The nondeterminism of an MDP is resolved with the help of a scheduler.
\begin{definition}
	A scheduler $\sigma$ (also known as adversary or policy) of an MDP $\mathcal{M}$ is a function mapping every finite path $\omega_{fin}$ onto an action $a\in A(last(\omega_{fin}))$ where $last(\omega_{fin})$ denotes the last state of $\omega_{fin}$.
\end{definition}

The scheduler $\sigma$ specifies the next action for each finite path. The behavior of an MDP $\mathcal{M}$ under a given scheduler $\sigma$ is purely probabilistic and thus reduces to a DTMC $\mathcal{M}_\sigma$. We denote $\Sigma_{\mathcal{M}}$ as the (possibly infinite) set of all possible schedulers for $\mathcal{M}$. There is a one-to-one correspondence between the paths of the DTMC and that of the MDP.
%
%
%
%
%
%
%
%
%
%
%
%
%

\subsection{Probabilistic Model Checking}
With the MDP model defined above, we can then determine if a given MDP satisfies some specification $\phi$. For the specification we use the Probabilistic Computation Tree Logic (PCTL) \cite{rutten2004mathematical} that can be seen as a probabilistic extension of Computation Tree Logic (CTL). The PCTL formulas we consider in this paper are the time-bounded weak-safety fragment \cite{chadha2010counterexample}. 
\begin{definition}
	The syntax of a weak-safety PCTL is defined as\\
	$\phi::=true~|~a~|~\phi\vee\wedge\phi~|~\mathcal{P}_{\leq p}[X\psi]~|~\mathcal{P}_{\leq p}[\psi~\mathcal{U}^{\leq k}~\psi]$;\\
	$\psi::=true~|~a~|~\psi\vee\wedge\psi~|~\neg(\mathcal{P}_{\leq p}[X\psi])~|~\neg(\mathcal{P}_{\leq p}[\psi~\mathcal{U}^{\leq k}~\psi])$,\\
	where $\phi$ is the weak safety fragment and $\psi$ is the strict liveness fragment, $a\in AP$, $p\in[0,1]$, $X$ for "next", $\mathcal{U}^{\leq k}$ for "bounded until" with $k\in\mathbb{N}$.
\end{definition}
Satisfaction of a PCTL formula $\mathcal{P}_{\leq p}(\varphi)$ means that the probability of satisfying $\varphi$ fulfills the comparison $\leq p$ \emph{over all possible schedulers}. 
The purpose of the probabilistic model checking is to give a Boolean answer to $\mathcal{M}\models\mathcal{P}_{\leq p}(\varphi)$.
\subsection{L* learning algorithm}
The L* learning algorithm  \cite{angluin1987learning} is one of widely used online learning algorithms for regular languages. It basically learns a \emph{minimal} Deterministic Finite Automata (DFA) that accepts an unknown regular language $\mathcal{L}$ by interacting with a \emph{teacher}. The definitions of DFA and language are described below.
\begin{definition}
	A \emph{finite automaton} (FA) is a tuple $\mathcal{A}=(Q,\Sigma,\delta,Q_0,F)$ where $Q$ is a finite set of states, $\Sigma$ is a finite alphabet of actions, $\delta:Q\times \Sigma\rightarrow 2^Q$ is a transition function, $Q_0\subseteq Q$ is a set of initial states and $F\subseteq Q$ is a set of \emph{accepting sates}. $\mathcal{A}$ is called \emph{deterministic} (DFA) if $|Q_0|\leq 1$ and $|\delta(q,a)|\leq 1$ for all states $q\in Q$ and $a\in\Sigma$.
\end{definition}
\begin{definition}
	A word $s=a_1a_2...a_n$ where $a_i\in\Sigma$ is accepted by an FA $\mathcal{A}=(Q,\Sigma,\delta,Q_0,F)$ if it induces a run $q_0q_1,...,q_n$ on FA where $q_i\in Q$ and $q_n\in F$. The collection of all finite words $S\in\Sigma^*$ accepted by $\mathcal{A}$ is called the language accepted by $\mathcal{A}$ denoted as $\mathcal{L}(\mathcal{A})$.
\end{definition}
\indent The teacher answers two types of questions, namely the \emph{membership} query and the \emph{conjecture}. The membership query returns true if some words belong to the target language and false otherwise. For the conjecture, if the hypothesized DFA accepts the target language, it will return true and the learning process is finished. Otherwise, the teacher must return a counterexample to illustrate the difference between the conjectured DFA and the unknown plant. 

During the learning process, the L* algorithm maintains a so-called \emph{observation table} $(S,E,T)$ where $S\subseteq\Sigma^*$ is a set of prefixes, $E\subseteq\Sigma^*$ is a set of suffixes and $T:(S\cup S\cdot\Sigma)\times E\rightarrow\{0,1\} $. $\Sigma^*$ denotes finite traces from the alphabet set $\Sigma$ 
For $s\in S\cup S\cdot\Sigma $, $e\in E$, if $s\cdot e\in\mathcal{L}$, then $T(s,e)=1$ and $T(s,e)=0$ if $s\cdot e\notin \mathcal{L}$. The L* algorithm will always keep the observation table \emph{closed and consistent} as defined in \cite{angluin1987learning}. It is guaranteed to learn a minimal DFA with $(|\Sigma|n^2 + n \log m)$ membership queries and
at most $n-1$ conjectures, where $n$ is the number
of states in the final DFA and $m$ is the length of the
longest counterexample in conjectures \cite{angluin1987learning}. Due to space limit we will not elaborate the details here.

	\section{Problem formulation}
	\label{sec:problem_formulation}
	Our target system model is an MAS consists of $N$ agents modeled as MDPs $\mathcal{M}_i,i\in[1,N]$. Define $\mathcal{M}$ as the composition of $N$ subsystems such that $\mathcal{M}=\mathcal{M}_1||\mathcal{M}_2||...||\mathcal{M}_N$. 
	
	The parallel composition of MDPs are defined as below:
	\begin{definition}
		Given two Markov decision processes $\mathcal{M}_1=(S_1,s^1_0,A_1,T_1,L_1)$ and $ \mathcal{M}_2=(S_2,s^2_0,A_2,$ $T_2,L_2)$, the parallel composition of $\mathcal{M}_1$ and $\mathcal{M}_2$ is an MDP $\mathcal{M}=\mathcal{M}_1||\mathcal{M}_2=(S_1\times S_2,s^1_0\times s^2_0,A_1\cup A_2,T,L)$ where $L(s_1,s_2)=L_1(s_1)\cup L_2(s_2)$ and  
		\begin{itemize}
			\item	$T((s_1,s_2),a,(s_1',s_2'))=T_1(s_1,a,s_1')T_2(s_2,a,s_2')$ if $a\in A_1\cap A_2$ and both $T_1(s_1,a,s_1')$ and $T_2(s_2,a,s_2')$ are defined, or
			\item $T((s_1,s_2),a,(s_1',s_2))=T_1(s_1,a,s_1')$ if $a\in A_1\backslash A_2$ and $T_1(s_1,a,s_1')$ is defined, or
			\item $T((s_1,s_2),a,(s_1,s_2'))=T_2(s_2,a,s_2')$ if $a\in A_2\backslash A_1$ and $T_2(s_2,a,s_2')$ is defined.
		\end{itemize}
	\end{definition}
	
	The probabilistic specification is given as $\phi=\mathcal{P}_{\leq p}(\varphi)$ where $\varphi=\phi_1\mathcal{U}^{\leq k}\phi_2$. It is a time-bounded weak-safety PCTL to guarantee certain properties. The reason why we focus on timed bounded formulas is that for many practical systems, because of power constraint and practical limitation like memory, the observation duration is often of limited time. Our aim is to determine if $\mathcal{M}\models\phi$. If not, we need to synthesize local supervisors $\mathcal{K}_i,i\in[1,N]$ that dynamically disable certain actions $a_i\in A_i$ so that the controlled system $\mathcal{M}$ can satisfy the specification $\phi$. 
	
	In particular, the supervisors $\mathcal{K}_i,i\in[1,N]$ is the form of DFAs with the alphabet set being given as $\Sigma_i=S_i\times A_i$. we define another notion of parallel composition $||_{sup}$ to illustrate how the supervisors regulate $\mathcal{M}$.
	
	\begin{definition}
		Given an MDP $\mathcal{M}=(S,s_0,A,T,L)$, and  the DFA $\mathcal{K}=\{Q,\Sigma,\delta,q_0,Q_m\}$, the parallel composition between $\mathcal{M}$ and $\mathcal{K}$ is an MDP $\mathcal{P}=\mathcal{M}||_{sup}\mathcal{K}$:
		\begin{itemize}
			\item $S^{\mathcal{P}}=\{(s,q)|s\in S,~q\in Q\}$ is a finite set of states;
			\item $(s_0,q_0)$ is the initial state;
			\item $A^{\mathcal{P}}=A$ is a finite set of actions;
			\item $T^{\mathcal{P}}((s,q),a,(s',q')):=T(s,a,s')$, if $\delta(q,sa)=q'$;
		\end{itemize}
	\end{definition} 
	
	There is a one-to-one correspondence between a path from $\mathcal{P}$, $\omega'=(s_0,q_0)\xrightarrow{a_0}(s_1,q_1)\xrightarrow{a_1}(s_2,q_2)...(s_i,q_i)\xrightarrow{a_i}(s_{i+1},q_{i+1})$ to the path in the original MDP $\mathcal{M}$, $\omega=s_0\xrightarrow{a_0}s_1\xrightarrow{a_1}s_2...s_i\xrightarrow{a_i}s_{i+1}$. From the definition of $||_{sup}$, only those transitions that are defined on both $\mathcal{M}$ and $\mathcal{K}$ are enabled, therefore the supervisor controls what actions can be allowed on each state. We assume that any $a_i\in A_i,\forall i\in[1,N]$ is controllable meaning that we could disable or enable it if we want. Note that in our previous results \cite{wu2015counterexample,wu2016counterexample},  the local supervisor's alphabet $\Sigma_i=A_i$. From the definition of the parallel composition $||$, it can be seen that it only keeps track of the action sequence but not state sequence. So our previous supervisor eliminates all the path with the same action sequence while in this paper, the supervisor keeps track of the state as well and thus is less conservative.
	
	Formally, our permissive supervisor synthesis problem can be formulated as follows:
	\begin{problem}{\textbf{1}}
		Given an MAS consisting of $N$ system models $\mathcal{M}_i=(S_i,s_{0,i},A_i,T_i,$ $L_i),~i=1,...,N$ and a probabilistic specification $\phi$ represented by PCTL such that $\mathcal{M}\not\models\phi$, where $\mathcal{M}=\mathcal{M}_1||\mathcal{M}_2||...||\mathcal{M}_N$. The supervisor synthesis problem is to automatically synthesize $N$ local supervisors $\mathcal{K}_i$ such that $\widetilde{\mathcal{M}}\models\phi$, where $\widetilde{\mathcal{M}}=\widetilde{\mathcal{M}_1}||\widetilde{\mathcal{M}_2}||...||\widetilde{\mathcal{M}_N}$ and $\widetilde{\mathcal{M}_i}=\mathcal{M}_i||_{sup}\mathcal{K}_i,~i=1,...,N$.
	\end{problem}
\section{Counterexample-guided supervisor synthesis framework}	\label{sec:CGSSF}
In this section, we will introduce the proposed counterexample guided supervisor synthesis framework. But before that, we need to first introduce the counterexamples in an MDP model. 
\subsection{Counterexamples}
\indent Counterexamples are one of the most important features in model checking which illustrates how one model violates certain property. They are the key ingredients in counterexample guided synthesis \cite{lin2011counterexample} and abstraction refinement \cite{clarke2003counterexample}. The form of counterexample varies by the checked formula and the given model. In non-probabilistic systems, for example safety properties like \emph{bad things should never happen} can be simply refuted by a single path that leads to a bad state. However, in probabilistic models, counterexample generation is non-trivial and is an active research area. The main difficulty is that, while one single path may be valid by obeying the probability bound under given scheduler, a collection of the finite paths may violate the probabilistic specification simply because their accumulative probability exceeds the desired bound.

For MDP models, since only with given scheduler can the MDP have probability measure, our first step is to find out the scheduler under which the specification $\phi$ is violated. In this case the MDP will be reduced to a DTMC. Then several results exist to generate counterexample in DTMC, for example, $k$ shortest path algorithm \cite{abraham2014counterexample} or minimal critical subsystem generation from \cite{wimmer2012minimal}.

\subsection{Single Agent Case}

A summary of our proposed supervisor synthesis framework is illustrated in Algorithm \ref{alg:supersyn}. Given an MDP model $\mathcal{M}$ and a probabilistic specification $\phi$, 
there are two tasks to be done in each iteration where $\mathcal{K}^i$ and $\widetilde{\mathcal{M}}^i$ are the resulting supervisor and the composed system respectively after the $i$-th iteration. 
\begin{figure}[!t]
	\removelatexerror
	\begin{algorithm}[H]
		\SetKwData{Left}{left}\SetKwData{This}{this}\SetKwData{Up}{up}
		\SetKwFunction{Union}{Union}\SetKwFunction{FindCompress}{FindCompress}
		\SetKwInOut{Input}{input}\SetKwInOut{Output}{output}
		\Input{An MDP model $\mathcal{M}$, probabilistic specification $\phi=\mathcal{P}_{\leq p}(\varphi)$}
		\Output{$\mathcal{K}^i$ such that $\widetilde{\mathcal{M}}^i=\mathcal{M}||_{sup}\mathcal{K}^i\models\phi$ or false if there does not exist such supervisor. }
		\BlankLine
		\nl $i\leftarrow 0$, $\mathcal{K}^i=$ DFA $\mathcal{G}$ such that
		$\mathcal{L}(\mathcal{G})=\Sigma^*$, $\widetilde{\mathcal{M}^i}=\mathcal{M}$;\
		
		\nl \While{true} {
			\nl \uIf{ModelChecking($\widetilde{\mathcal{M}^i},\phi$)=false}{
				\nl \If{i=0 \& $P_{min}(\varphi)>p$}{\Return false;}
				
				\nl ${\widetilde{\mathcal{M}}_{\sigma_i}}\leftarrow$ the counterexample DTMC; \
				
				\nl $\pi^i\leftarrow$ HSP(${\widetilde{\mathcal{M}}_{\sigma_i}},\phi$);

			}
			\uElseIf{Positive counterexample exists}{
				\nl $\pi^i\leftarrow$ Positive counterexample; \
			}
			
			\uElse{Break;}
			\nl $\mathcal{K}^{i+1}\leftarrow$ L* ($\pi^i$), $\widetilde{\mathcal{M}}^{i+1}=\mathcal{M}||_{sup}\mathcal{K}^{i+1}$\;
			
			\nl $i\leftarrow i+1$;
		}
		\nl \Return $\mathcal{K}^i$\;
		
		\caption{SPVSYN($\mathcal{M},\phi$)}\label{alg:supersyn}		
	\end{algorithm}
\end{figure}
%
%
%
%
%
%
%
%

\noindent i) Verification as shown from line 3 to line 7 of Algorithm \ref{alg:supersyn}. In the first model checking ($i=0$), if $P_{min}<=p$, we can extract the strategy $\sigma_{min}$ as a DFA for later use.  Given $\widetilde{\mathcal{M}}^i$ and the probabilistic property $\phi=\mathcal{P}_{\leq p}(\varphi)$ where $\varphi=\phi_1U^{\leq t}\phi_2$, this phase first checks whether $\widetilde{\mathcal{M}}^i \models \phi$. If not, a DTMC $T^i$ induced from a scheduler will be returned as the counterexample to select the counterexample path $\pi^i$. Else if there is positive counterexample, it will be saved to $\pi^i$. If no counterexample in both cases, then we are done.

\noindent ii) Learning based synthesis. A synthesis-specific implementation of the L* learning algorithm is proposed here. In this stage, to be able to answer the membership query, we will view $\mathcal{M}$ as a finite automaton (FA) by ignoring its probabilistic attributes and leaving only its states, transition relation and labeling. Formally, we will reduce $\mathcal{M}=(Q,q_0,A,Steps,L)$ into a labeled FA $\mathcal{M}_D=(Q,A,\delta,q_0,F,L)$ where for all $q,q'\in Q$ and $a\in A$, $\delta(q,a)=q'$ if and only if there is an action-distribution pair $(a,\mu)\in Steps(q)$ and $\mu(q')>0$. The accepting states $F=Q$. Since we only consider a bounded time $k$, any path longer than $k$ doesn't matter to the supervisor and we  answer them with $1$.

\indent We then describe the learning process in more detail. It makes use of the L* learning as a template algorithm to learn the unknown supervisor $\mathcal{K}$. The teacher to membership query will be $\mathcal{M}_D$ which tells if some finite state action path belongs to it or not. After some membership queries, $\mathcal{K}^i$ will be generated as in line 7 in Algorithm \ref{alg:supersyn}. The conjecture query is answered by probabilistic model checking as shown in line 3 to line 6 in Algorithm \ref{alg:supersyn} for negative counterexamples (those should be disallowed but currently allowed by the supervisor) and $\mathcal{M}_D$ together with the $\sigma_{min}$ for positive counterexamples (those should be allowed but currently absent in the supervisor). If the probabilistic model checking returns false, it will return DTMC $T^i$ as in line 5 and a state action path $\pi^i$ will be selected as the counterexample for the conjecture as in line 6. If positive counterexample exists, it will be passed to L* learning as in line 7. In line 8, the L* learning algorithm will update its observation table accordingly and run another round of membership query to conjecture the supervisor. To be permissive, for $i=0$ every possible path is allowed. After that, $\mathcal{K}^i$ where $i\geq1$ will be conjectured iteratively until neither positive nor negative counterexample is generated. Then the supervisor is obtained by eliminating all its transitions to the non-accepting states.

A notable difference of our modified L* learning algorithm from the traditional one is that in our framework, the answer to the membership query might be changed by the counterexamples returned from the conjecture, while in traditional L* learning the answer to the membership query will remain the same throughout the learning process. The reason is that in our case, as we are looking for a permissive supervisor, when some finite path $\omega$ is accepted by the $\mathcal{M}_D$, our membership query will return true and our supervisor will allow it to happen first. But it could be the case that $\omega\in\pi^i$ at some iteration $i$ later which means that it is actually found to be a counterexample path and should be eliminated. Then the answer to this membership query will be changed from true to false and the observation table will be again checked for closedness and consistency and may bring several additional rounds of membership queries.

\subsubsection{Counterexample selection}\label{subsubsec:CE}
When $\widetilde{\mathcal{M}}^i \not\models \phi$ where $\phi=\mathcal{P}_{\leq p}(\varphi)$, it will return the DTMC $\widetilde{\mathcal{M}}_{\sigma_i}$ induced by the optimal scheduler $\sigma_i$ on $\widetilde{\mathcal{M}}^i$ such that $P_{\widetilde{\mathcal{M}}_{\sigma_i}}(\varphi)>p$. Then this stage is responsible for selecting a particular counterexample path $\pi^i$ to eliminate. We could apply the hop-constrained shortest path (HSP) algorithm
\cite{han2009counterexample} to find the path with the largest probability that is not induced by $\sigma_{min}$. Since $\sigma_{min}$ is not a counterexample strategy, such path is guaranteed to exist. 
The resulting $\pi^i$ will then be passed to the third stage to answer the conjecture in our modified L* learning algorithm. Note that in our previous work \cite{wu2015counterexample,wu2016counterexample}, we proposed to apply the hop-constrained $k$ shortest path (HKSP) algorithm to find the smallest $k$ paths so that eliminating them will make sure the DTMC will not be counterexample any more. The complexity of HKSP is higher than the HSP since HKSP has to search for $k$ as well. Additionally, any state action path with length no larger than $t$ that exists in the original MDP but not in the supervised MDP and none of the counterexample paths is its prefix can serve as a positive counterexample. Such path can be found by computing the difference between the regular languages generated by the underlying finite automata (with the state action pair as alphabet) of the MDP and the supervised MDP. 

\subsubsection{Correctness and termination}
\begin{assumption}\label{assumption:support}
	For each $s\in S$ from $\mathcal{M}$, the support for every distribution $\mu_s^a$ where $\mu_s^a(s') = T(s,a,s')$ are the same.
\end{assumption}
Each iteration of synthesis flow will terminate because $\mathcal{K}$ is a DFA and L* learning algorithm learns a minimal DFA in a finite number of queries \cite{angluin1987learning}. Theorem \ref{thm:correct} proves that the result of supervisor synthesis is correct by design and always terminates.
%
\begin{thm}\label{thm:correct}
	If the supervisor exists, the synthesized supervisor $\mathcal{K}$ parallel composes with original MDP $\mathcal{M}$ satisfies the specification $\phi$ and the whole synthesis process also terminates in finite iterations and the supervisor is nonblocking. 
\end{thm}
\begin{proof}
	If $\mathcal{M}\models\phi$, the synthesis terminates when $i=0$ and the theorem trivially holds. Otherwise, for probabilistic property in the form of $\phi=\mathcal{P}_{\leq p}(\varphi)$ where $\varphi=\phi_1U^{\leq t}\phi_2$, the number of possible schedulers is upper bounded by $|SA|^t$. In each iteration, we eliminate a class of schedulers that shares the same decision on the counterexample path. Therefore, the number of iterations $D\leq |SA|^t<\infty$. That is, the synthesis terminates in a finite steps and from the termination condition, we can guarantee that $\mathcal{M}||\mathcal{K}\models\phi$. For nonblocking, from our learning process, $\sigma_{min}$ will always be preserved in the supervisor. Furthermore, from Assumption \ref{assumption:support}, we know $\sigma_{min}$ will be defined for all possible paths.
\end{proof}
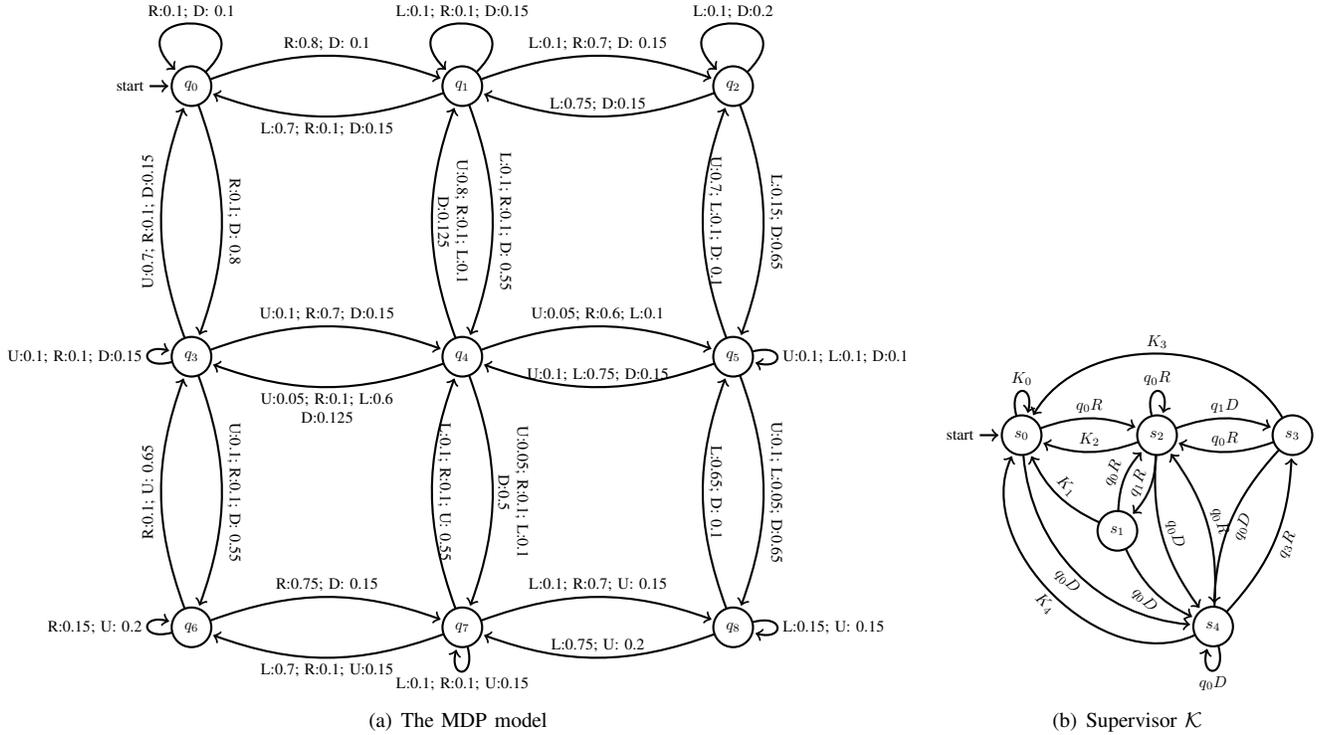
\begin{figure*}[ht]
	
	\centering
	\subfigure[The MDP model]{
		\begin{tikzpicture}[shorten >=1pt,node distance=6cm,on grid,auto, bend angle=20, thick,scale=0.6, every node/.style={transform shape}]
		\node[state,initial] (s_0)   {$q_0$}; 
		\node[state] (s_1) [right= of s_0] {$q_1$}; 		
		\node[state] (s_2) [right= of s_1] {$q_2$}; 
		\node[state] (s_3) [below =of s_0] {$q_3$}; 	
		\node[state] (s_4) [below= of s_1] {$q_4$}; 
		\node[state] (s_5) [below =of s_2] {$q_5$}; 
		\node[state] (s_6) [below =of s_3] {$q_6$}; 	
		\node[state] (s_7) [below= of s_4] {$q_7$}; 
		\node[state] (s_8) [below =of s_5] {$q_8$}; 
		
		\path[->] 
		
		(s_0) edge  [bend left=20] node  [pos=0.5, sloped, above][align=center][fontscale=0.01]{R:0.8; D: 0.1} (s_1)
		(s_0) edge [pos=0.5, loop, above=0.1] node [align=center] {R:0.1; D: 0.1} (s_0)
		(s_0) edge  [bend left=20] node  [pos=0.5, sloped, above][align=center]{R:0.1; D: 0.8} (s_3)

		(s_1) edge [bend left=20]  node [pos=0.5, sloped, above][align=center]{ L:0.1; R:0.7; D: 0.15} (s_2)
		
		(s_1) edge [bend left=20]  node[pos=0.5, sloped, below][align=center]{ L:0.7; R:0.1; D:0.15} (s_0)
		
		(s_1) edge [bend left=20]  node [pos=0.5, sloped, above][align=center]{ L:0.1; R:0.1; D: 0.55} (s_4)
		
		(s_1) edge [pos=0.5, loop, above=0.1] node [align=center] {L:0.1; R:0.1; D:0.15} (s_1)

		(s_2) edge  [bend left=20] node  [pos=0.5, sloped, above][align=center][fontscale=0.01]{L:0.75; D:0.15} (s_1)
		(s_2) edge [pos=0.5, loop, above=0.1] node [align=center] {L:0.1; D:0.2} (s_2)
		(s_2) edge  [bend left=20] node  [pos=0.5, sloped, above][align=center]{L:0.15; D:0.65} (s_5)
		
		(s_3) edge [bend left=20]  node [pos=0.5, sloped, above][align=center]{ U:0.7; R:0.1; D:0.15} (s_0)
		
		(s_3) edge [bend left=20]  node[pos=0.5, sloped, above][align=center]{ U:0.1; R:0.7; D:0.15} (s_4)
		
		(s_3) edge [bend left=20]  node [pos=0.5, sloped, above][align=center]{ U:0.1; R:0.1; D: 0.55} (s_6)
		
		(s_3) edge [pos=0.5, loop left] node [align=center] {U:0.1; R:0.1; D:0.15} (s_3)
		
		(s_4) edge [bend left=20]  node [pos=0.5, sloped, above][align=center]{ U:0.8; R:0.1; L:0.1\\D:0.125 } (s_1)
		
		(s_4) edge [bend left=20]  node[pos=0.5, sloped, below][align=center]{ U:0.05; R:0.1; L:0.6\\D:0.125} (s_3)
		
		(s_4) edge [bend left=20]  node [pos=0.5, sloped, above][align=center]{ U:0.05; R:0.6; L:0.1} (s_5)
		
		(s_4) edge [bend left=20]  node [pos=0.5, sloped, above][align=center]{ U:0.05; R:0.1; L:0.1\\D:0.5} (s_7)
		
		(s_5) edge [bend left=20]  node [pos=0.5, sloped, above][align=center]{ U:0.7; L:0.1; D: 0.1} (s_2)
		
		(s_5) edge [bend left=20]  node[pos=0.5, sloped, above][align=center]{ U:0.1; L:0.75; D:0.15} (s_4)
		
		(s_5) edge [bend left=20]  node [pos=0.5, sloped, above][align=center]{ U:0.1; L:0.05; D:0.65} (s_8)
		
		(s_5) edge [pos=0.5, loop right] node [align=center] {U:0.1; L:0.1; D:0.1} (s_5)
		
		(s_6) edge  [bend left=20] node  [pos=0.5, sloped, above][align=center]{R:0.1; U: 0.65} (s_3)
		(s_6) edge [pos=0.5, loop left] node [align=center] {R:0.15; U: 0.2} (s_6)
		(s_6) edge  [bend left=20] node  [pos=0.5, sloped, above][align=center]{R:0.75; D: 0.15} (s_7)
		
		(s_7) edge [bend left=20]  node [pos=0.5, sloped, above][align=center]{ L:0.1; R:0.1; U: 0.55} (s_4)
		
		(s_7) edge [bend left=20]  node[pos=0.5, sloped, below][align=center]{ L:0.7; R:0.1; U:0.15} (s_6)
		
		(s_7) edge [bend left=20]  node [pos=0.5, sloped, above][align=center]{ L:0.1; R:0.7; U: 0.15} (s_8)
		
		(s_7) edge [pos=0.5, loop below] node [align=center] {L:0.1; R:0.1; U:0.15} (s_7)

		(s_8) edge  [bend left=20] node  [pos=0.5, sloped, above][align=center]{L:0.75; U: 0.2} (s_7)
		(s_8) edge [pos=0.5, loop right] node [align=center] {L:0.15; U: 0.15} (s_8)
		(s_8) edge  [bend left=20] node  [pos=0.5, sloped, above][align=center]{L:0.65; D: 0.1} (s_5)		
		;
		
		\end{tikzpicture}
	}
	\subfigure[Supervisor $\mathcal{K}$]{
		\begin{tikzpicture}[shorten >=1pt,node distance=3cm,on grid,auto, bend angle=20, thick,scale=0.6, every node/.style={transform shape}]
		\node[state,initial] (s_0)   {$s_0$};
		
		\node[state] (s_2) [right= of s_0] {$s_2$}; 
		\node[state] (s_1) [below right =of s_0] {$s_1$}; 		
		\node[state] (s_3) [right= of s_2] {$s_3$}; 
		\node[state] (s_4) [below right=  of s_1] {$s_4$}; 
		
		\path[->] 
		
		(s_0) edge [loop above]  node [pos=0.5, sloped, above]{$K_0$} (s_0)
		
		(s_0) edge [bend left=20]  node [pos=0.5, sloped, above]{$q_0R$} (s_2)
		
		(s_0) edge [bend right=40]  node [pos=0.5, sloped, below]{$q_0D$} (s_4)
		
		(s_1) edge [bend left=20] node [pos=0.5, sloped, above][align=center]{$K_1$} (s_0)
		
		(s_1) edge [bend left=20]  node [pos=0.5, sloped, above]{$q_0R$} (s_2)
		
		(s_1) edge [bend right=20]  node [pos=0.5, sloped, below]{$q_0D$} (s_4)
		
		(s_2) edge [bend left=20] node [pos=0.5, sloped, above][align=center]{$K_2$} (s_0)
		
		(s_2) edge [bend left=20]  node [pos=0.5, sloped, above]{$q_1R$} (s_1)

		(s_2) edge [loop above]  node [pos=0.5, sloped, above]{$q_0R$} (s_2)

		(s_2) edge [bend left=20]  node [pos=0.5, sloped, above]{$q_1D$} (s_3)
		
		(s_2) edge [bend right=20]  node [pos=0.5, sloped, above]{$q_0D$} (s_4)
		
		(s_3) edge [bend right=60] node [pos=0.5, sloped, above][align=center]{$K_3$} (s_0)
		
		(s_3) edge [bend left=20]  node [pos=0.5, sloped, above]{$q_0R$} (s_2)
		
		(s_3) edge [bend right=20]  node [pos=0.5, sloped, below]{$q_0D$} (s_4)
		
		(s_4) edge [bend left=70] node [pos=0.5, sloped, below][align=center]{$K_4$} (s_0)
		
		(s_4) edge [bend right=20]  node [pos=0.5, sloped, above]{$q_0R$} (s_2)
		
		(s_4) edge [bend right=20]  node [pos=0.5, sloped, below]{$q_3R$} (s_3)

		(s_4) edge [loop below]  node [pos=0.5, sloped, below]{$q_0D$} (s_4)	
		;      
		\end{tikzpicture}
	}
	\caption{The MDP model to synthesis and the resulting supervisor}
	\label{fig:single}
\end{figure*}

\subsubsection{Complexity}
We define the size of an MDP $\mathcal{M}, S_{\mathcal{M}}$ as the total number of nondeterministic choices. In each iteration $i$,  model checking has complexity polynomial with the size of the composed MDP $\widetilde{M}^i$ which is upper bounded by $|SA|^t$ since we only consider all the possible actions in $t$ steps. The counterexample selection algorithm has complexity exponential with $t$. It is also linear with the length of the PCTL formula $L_f$ where $L_f$ is defined to be the number of logical and temporal operators in the formula. The complexity for checking positive counterexample is bounded by $|SA|^t$. The number of iterations in the worst case is exponential with the time horizon $t$ since there are at most $|SA|^t$ number of possible policies. As to L* learning, the number of the states in the resulting DFA is upper bounded by $|SA|^t$ and the longest counterexample is bounded by $t$.  Therefore to sum up, in the worst case our algorithm has a complexity polynomial to $|S|$ and $|A|$, exponential with $t$ and linear with $L_f$. 
\subsubsection{Illustrative example}

We give an illustrative example to show our learning based permissive supervisor synthesis framework. The system is as shown in Fig. \ref{fig:single} and we ignore the labeling for simplicity. It is essentially a robot with initial position at $s_0$ moving in a grid space. The action set is \{U,D,L,R\} representing moving up, down, left and right. The number after the colon is the transition probability. For $s_4$, we didn't draw its self loop transition due to the space limitation. But its transition probabilities can be easily inferred. 

The probabilistic specification is $\phi=P_{\leq 0.6}(\varphi)$ where $\varphi=trueU^{\leq 5}q_5$. To find the supervisor, we developed our own software toolkit which first does the model checking and gets the induced DTMC based on counterexample strategy. Then both DTMC and the specification (on this DTMC) is sent to COMICS to find the counterexample path which is then fed into L* learning using libalf. Once the conjectured supervisor is found, we parallel it with the original plant to see if the specification is satisfied. As in Algorithm \ref{alg:supersyn}, it runs iteratively until no counterexample is found.

For this particular example, the synthesis process runs for 4 iterations and the maximum probability to satisfy $\varphi$ at the termination is $0.573053$. The resulting supervisor has 5 states, as shown in Fig \ref{fig:single} where
\begin{itemize}
	\item $K_0 = \{q_ia|\forall i>0,(q_i,a,*)\in T\}$
	\item $K_1 = \{q_ia|\forall i>0,(q_i,a,*)\in T\}\backslash\{q_2D\}$
	\item $K_2 = \{q_ia|\forall i>1,(q_i,a,*)\in T\}\cup\{q_1L\}$
	\item $K_3 = \{q_ia|\forall i>0,(q_i,a,*)\in T\}\backslash\{q_4R\}$
	\item $K_4 = \{q_ia|\forall i>0,(q_i,a,*)\in T\}\backslash\{q_3R\}$

\end{itemize}
We also tested another case setting the probability threshold to $0.5$ in the specification, this time the synthesis process run for 14 iterations, the maximum probability at the termination is $0.498412$ and the resulting supervisor has 18 states. It is expected since the lower probability threshold can result in more counterexample strategies and thus more iterations are needed to eliminate them.  

	\subsection{Multi-Agent Case}
	
	For simplicity, we assume that $N=2$ but our framework can be readily extended to the case where $N>2$.
	\subsubsection{Learning based supervisor synthesis}
	Algorithm \ref{alg:nsupersyn} describes the flow of our framework. 
	The whole work flow is similar to the single agent case, but with several key differences as described below. 
	%
	
	The first difference is the  model checking. Given an MAS consisting of $N$ system models $\mathcal{M}_i=(S_i,s^i_0,A_i,T_i,$ $L_i),~i=1,...,N$, the initial local supervisors $\mathcal{K}_i$ are trivial ones which allow every action. Then to avoid state space explosion, we will apply compositional model checking techniques to be introduced in Section \ref{sec:CMC} to see if the given specification $\phi$ can be satisfied by the uncontrolled system. If the answer is yes and there is no positive counterexample, then we are done. Otherwise, counterexamples will be returned. 
	
	
	The second difference is the subsystem selection. After getting the returned counterexamples, it is time to determine which subsystem is at fault to such violation. The detail will be discussed in Section \ref{section:CSS}. After identifying which subsystem $\mathcal{M}_k$ actively causes the violation of the specification, its local supervisor will be refined. Other subsystem's supervisor will be refined based on the positive counterexamples. The refining process is based on L* learning algorithm with the returned counterexamples. 
	%
	\begin{figure}[!t]
		\removelatexerror
		\begin{algorithm}[H]
			\SetKwData{Left}{left}\SetKwData{This}{this}\SetKwData{Up}{up}
			\SetKwFunction{Union}{Union}\SetKwFunction{FindCompress}{FindCompress}
			\SetKwInOut{Input}{input}\SetKwInOut{Output}{output}
			\Input{  $\mathcal{M}_i=(S_i,s_{0,i},A_i,T_i, L_i),~i=1,...,N$, probabilistic specification $\phi=\mathcal{P}_{\leq p}(\varphi)$}
			\Output{Local supervisors $\mathcal{K}_i,i\in[1,N]$ such that $\widetilde{\mathcal{M}}\models\phi$, where $\widetilde{\mathcal{M}}=\widetilde{\mathcal{M}_1}||\widetilde{\mathcal{M}_2}||...||\widetilde{\mathcal{M}_N}$ and $\widetilde{\mathcal{M}_i}=\mathcal{M}_i||_{sup}\mathcal{K}_i,~i=1,...,N$. }
			\BlankLine
			\nl $j\leftarrow 0$, $\mathcal{K}^0_i=$DFA $\mathcal{G}_i$ such that $\mathcal{L}(\mathcal{G}_i)=\Sigma_i^*$,$\widetilde{\mathcal{M}_i^0}\leftarrow\mathcal{M}_i$;\		
			

			\nl \While{true} {	
				\nl $k=-1$\;
				\nl \uIf
				{ModelChecking($\widetilde{\mathcal{M}^j},\phi$)=false}{
					$\pi^j\leftarrow$ counterexample path  \;
					
					$k\leftarrow$SELECT($\mathcal{M},\pi^j$) from \cite{lin2011counterexample}\;
					$\pi_k^j\leftarrow$ projection from $\pi^j$\;
				}

				\For{$i\in[1,N]$}{\uIf{$i\neq k$}{$\pi^j_i\leftarrow$ positive counterexample\;}
					\nl { $\mathcal{K}^{j+1}_{i}\leftarrow$ L* Learning($\pi^j_i$)\;}
					\nl $\widetilde{\mathcal{M}}^{j+1}_{i}\leftarrow\mathcal{M}_{i}||\mathcal{K}^{j+1}_{i}$;	}
				\uIf{Positive or negative counterexamples exist}{$j\leftarrow j+1$\;}
				\uElse{Break;}
				
			}
			\nl \Return {$\mathcal{K}^j_i,i\in[1,N]$}\;
			
			~\caption{N-SPVSYN($\mathcal{M},\phi$)}\label{alg:nsupersyn}			
		\end{algorithm}
	\end{figure}

	\subsubsection{Compositional model checking}\label{sec:CMC}
	With multiple agents in the system, the total state space grows exponentially with $N$ which makes the model checking computation prohibitively expensive. Therefore the conventional model checking methods will simply fail due to the high computation complexity. To prevent this, we propose to refer to compositional verification of the probabilistic systems in which the composed system is never built. The most well-known compositional technique is assume-guarantee reasoning (AGR). It has already been successfully applied to non-probabilistic model checking \cite{henzinger1998you} and recently has been extended to probabilistic systems \cite{hermanns2008probabilistic,chadha2010counterexample,komuravelli2012assume}. Here we use the asymmetric rules (ASYM) as shown below.\\
	\begin{equation}
	\begin{gathered}
	\underline{1:\mathcal{M}_1||A_b\models\phi~~~~~~ 2:\mathcal{M}_2\preceq A_b}\\
	\mathcal{M}_1||\mathcal{M}_2\models\phi
	\end{gathered}
	\end{equation}
	where $A_b$ is called the assumption for $\mathcal{M}_2$. Often the assumption is much smaller than the original system and the composition of $\mathcal{M}_1$ and $A_b$ could be much smaller to alleviate the computation burden. The key problem for ASYM is to find a proper assumption and the main idea is to use counterexamples to refine the assumptions until the ASYM can answer the model checking problem.  
	
	While our framework has the flexibility in choosing different ASYM algorithms, we require the ASYM to be sound and complete with which a proper assumption can always be found for model checking purpose by showing the satisfaction of given specification or returning real counterexamples that witness the violation for the original system. For example, in \cite{hermanns2008probabilistic}, probabilistic automaton is used as the assumption form; with interpolation \cite{henzinger2004abstractions} based refinement algorithm, a sound and complete counterexample-guided abstraction refinement (CEGAR) framework is constructed with counterexample in the form of state action paths. With real counterexamples returned from the ASYM, the distributed supervisor can be revised accordingly.

	\subsubsection{Counterexample and subsystem selection}\label{section:CSS}
	%
	%
	%
	\indent For probabilistic model checking of multiple plants, when the specification is not satisfied, the returned counterexamples belong to the composed system. After getting the counterexample paths, what's different from dealing with the single plant is that it is important to find out which subsystem \emph{actively} performs the last action of each paths and causes the specification violation.\\
	\indent To do this, inspired by \cite{lin2011counterexample}, we propose to partition the action set into \emph{active}, \emph{passive} and \emph{normal} actions. Formally, for each subsystem $i$, we categorize the action set $A_i$ into three disjoint sets $A_{i,a},A_{i,p}$ and $A_{i,n}$, namely active, passive and normal actions and we only consider the \emph{well-communicated systems} \cite{lin2011counterexample} as defined below.
	\begin{definition}
		Given an MAS $\mathcal{M}$ consisting of $N$ system models $\mathcal{M}_i=(S_i,s^i_0,A_i,T_i,$ $L_i),~i=1,...,N$, we say $\mathcal{M}$ is well communicated if the following hold:
		\begin{itemize}
			\item $\forall i\in [1,N]$, for each active action $a\in A_{i,a}$, there exists at least one corresponding passive action $a\in A_{j,p},j\neq i$. 
			\item $\forall i\in [1,N]$, for each passive action $a\in A_{i,p}$, there exists one and only one corresponding active action $a\in A_{j,a},j\neq i$.
		\end{itemize}
	\end{definition}
	\indent It is worth noting that such labeling only helps us to identify which subsystem is the cause of the violation. It is completely ignored in the model checking and supervisor synthesis stage. We adopt the SELECT algorithm (Algorithm 4) in \cite{lin2011counterexample}. Intuitively, we start from the end of the action sequence and find the subsystem that actively causes the last action to be performed.
	%
	%
	%
	In this case, we can just focus on which subsystem $\mathcal{M}_j$ performs the last active action of counterexample path $\omega$. Once we find the subsystem $\mathcal{M}_j$, by projecting $\omega$ to the local action set of $\mathcal{M}_j$, we get a local counterexample path $\omega_{\mathcal{M}_j}$. Then it can be seen that the problem is to learn a local supervisor $\mathcal{K}_j$. 
	For other subsystems that are not selected, their supervisors remain unchanged. After that,  we perform another round of model checking and refinement until no counterexample is generated. 
	\subsubsection{Computation Complexity, Correctness and Termination}
	In each iteration, for the framework we use in Section \ref{sec:agrexp}, the model checking has the complexity upper bounded by $\prod_{i=1}^N|S_iA_i|^t$. It is also linear with the length of the PCTL formula $L_f$. The number of iterations is finite and in the worst case is exponential with the time horizon $t$ and $N$ since there are at most $\prod_{i=1}^N|S_iA_i|^t$ number of possible policies. As to L* learning, the number of the states in the resulting DFA is upper bounded by $|SA|^t$ and the longest counterexample is bounded by $t$. Therefore in the worst case, our algorithm has a complexity exponential with $t,N$, polynomial with $S_{\mathcal{M}_i}$,$|S_i|$, $|A_i|$ and linear with $L_f$. However it is seen that in practice the compositional model checking rarely assume huge computation \cite{komuravelli2012assume}. So this complexity analysis is rather conservative.
	
	We then prove that the result of supervisor synthesis is correct by design and our framework always terminates. 

	\begin{thm}\label{thm:correct}
		The synthesized supervisors $\mathcal{K}_i$ parallel composes with original MDP $\mathcal{M}_i$ and the composed controlled system satisfies the specification $\phi$. The whole synthesis process also terminates in finite iterations. 
	\end{thm}
	\begin{proof}
		From the soundness and completeness of ASYM rule with CEGAR, we know that our compositional model checking will always terminate and produce the correct answer. If $\mathcal{M}\models\phi$, the synthesis terminates when $i=0$ and our theorem trivially holds. Otherwise in each iteration, counterexample selection algorithm will always terminate since we have only a finite number of agents. In the local supervisor synthesis, our modified L* learning algorithm will always terminate in a finite number of quires \cite{angluin1987learning}. So each iteration will terminate in finite time. Furthermore, the number of iterations is also finite since we are dealing with time bounded time properties. To sum up, the synthesis terminates in  finite steps and from the termination condition and the soundness of AGAR, we can guarantee that $\widetilde{\mathcal{M}}\models\phi$, where $\widetilde{\mathcal{M}}=\widetilde{\mathcal{M}_1}||\widetilde{\mathcal{M}_2}||...||\widetilde{\mathcal{M}_N}$ and $\widetilde{\mathcal{M}_i}=\mathcal{M}_i||\mathcal{K}_i,~i=1,...,N$.
	\end{proof}
	One of the conservativeness of the proposed method is that since we are using the compositional model checking , we are not able to get the minimum probability and the corresponding $\sigma_{min}$ like in the single agent setting. It could be the case that at least one of the supervisors becomes blocking. Such scenario may be avoided in the run time by not executing the actions that lead to the blocking state. It also could be the case that the minimum probability to satisfy the property $\varphi$ is larger than $p$. Then the model checking will keep on reporting counterexamples until the supervisors in one or more subsystems become blocking. 
	\subsubsection{An illustrative example}\label{sec:agrexp}
	%
	%
	%

	Here we give an example to show the integration of our framework with assume-guarantee reasoning for MDPs. While our framework does not limit the choice of assume-guarantee reasoning algorithm to reduce the state space of MDPs, the counterexample guided abstraction refinement method developed in \cite{hermanns2008probabilistic} is used for illustration purpose in this example. The system consists of two MDP models. The first one is as shown in Fig. \ref{fig:mpd1}.
	\begin{figure}[h]
		
		\centering
		
		\begin{tikzpicture}[shorten >=1pt,node distance=3cm,on grid,auto, bend angle=20, thick,scale=0.75, every node/.style={transform shape}]
		\node[state,initial above] (s_0)   {$s_0$}; 
		\node[state] (s_1) [right= of s_0] {$s_1$}; 		
		\node[state] (s_2) [left= of s_0] {$s_2$};

		\path[->] 
		
		(s_0) edge [bend left=30]  node [pos=0.5, sloped, above]{$a$} (s_1)

		(s_1) edge [bend left=30]  node [pos=0.5, sloped, above]{$a$} (s_0)
		
		(s_0) edge  [bend left=30] node [pos=0.5, sloped, above]{$b,c$} (s_2)	
		
		(s_2) edge [bend left=20]  node [pos=0.5, sloped, above]{$b,c$} (s_0)

		;      
		\end{tikzpicture}
		
		\caption{$\mathcal{M}_1$ model}
		\label{fig:mpd1}
	\end{figure}
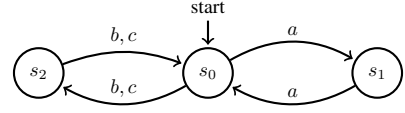
	
	The system model $\mathcal{M}_2$ is extended from the example discussed in \cite{hermanns2008probabilistic}. For $\mathcal{M}_1$, the state set is $S_1=\{0,1,...,N-1\}\times \{bad,!bad\}, N>2$ with initial state $\langle 0,!bad \rangle$; the action set is $A=\{a,b,c\}$. The transition function is given as follows. For $i<N-1$, 
	\begin{align*}
	T(\langle i,!bad \rangle,a,\langle i+1,!bad \rangle)=0.9,\\
	T(\langle i,!bad \rangle,a,\langle i,!bad \rangle)=0.1.
	\end{align*}  
	For $i=N-1$
	\begin{align*}
	T(\langle N-1,!bad \rangle,a,\langle N-1,!bad \rangle)=0.9,\\
	T(\langle N-1,!bad \rangle,a,\langle N-1,bad \rangle)=0.1.
	\end{align*}
	For action $b$, we have it define on the initial state with
	\begin{align*}
	T(\langle 0,!bad \rangle,b,\langle N-1,bad \rangle)=0.5,\\
	T(\langle 0,!bad \rangle,b,\langle 0,!bad \rangle)=0.5.
	\end{align*}

	For $\mathcal{M}_1$, $\{a,b\}\in A_{1,p}$ and $c\in A_{1,a}$. For $\mathcal{M}_2$, $\{a,b\}\in A_{2,a}$ and $c\in A_{2,p}$. It is not hard to verify that $\{\mathcal{M}_1,\mathcal{M}_2\}$ compose a well-communicated system. $L_2(\langle N-1,bad \rangle)=\{failure\}$ and the specification is given as $\mathcal{P}_{\leq 0.3}[true~\mathcal{U}^{\leq N-1}~failure]$.

	\begin{figure}
		\centering
		
		\begin{tikzpicture}[shorten >=1pt,node distance=4cm,on grid,auto, bend angle=20, thick,scale=0.75, every node/.style={transform shape}] 
		\node[state,initial] (q0)[align=right]  {$s^\#_0$};
		\node[state] (q1) [right=of q0,align=right] {$s^\#_1$};			
		\node[state] (q2) [below=of q1,align=right] {$s^\#_2$};    			
		\path[->]
		(q0) edge node {$a:0.9$} (q1)  
		edge  [bend right=30] node[pos=0.5, sloped, above] {$b:0.5;c:0.5$} (q2)
		edge [loop above, pos = 0.5,sloped, above] node {$a:0.1;~b:0.5;~c:0.5$} ()
		(q1) edge node {$a:0.9;~a:0.1$} (q2)  
		edge [loop above, pos = 0.5,sloped, above] node {$a:0.9;~a:0.1$} ()
		; 
		\node[text width=3cm] at (-1,-0.5) 
		{\{$i=0,!bad$\}};
		\node[text width=3cm] at (6,-0.5) 
		{\{$i<N,!bad$\}};
		\node[text width=3cm] at (6,-4) 
		{\{$i=N-1,bad$\}};
		\end{tikzpicture} 
		\caption{The initial abstract quotient automaton $\mathcal{M}^\#_2$}
		\label{fig:aM1}
		
	\end{figure}
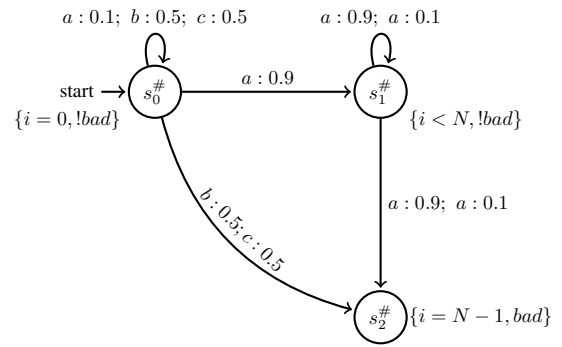
	
	%
	%
	%
	%
	%
	
	\begin{figure}
		\centering
		
		\begin{tikzpicture}[shorten >=1pt,node distance=4cm,on grid,auto, bend angle=20, thick,scale=0.75, every node/.style={transform shape}] 
		\node[state,initial] (q0)[align=right]  {$s^\#_0$};
		\node[state] (q1) [right=of q0,align=right] {$s^\#_1$};			
		\node[state] (q2) [below=of q0,align=right] {$s^\#_2$};    	
		\node[state] (q3) [below=of q1,align=right] {$s^\#_3$}; 		
		\path[->]
		(q0) edge node {$a:0.9$} (q1)  
		edge   node[pos=0.5, sloped, above] {$b:0.5;~c:0.5$} (q2)
		edge [loop above, pos = 0.5,sloped, above] node {$a:0.1;~b:0.5;~c:0.5$} ()
		(q1) edge node {$a:0.9$} (q3)  
		edge [loop above, pos = 0.5,sloped, above] node {$a:0.9;~a:0.1$} ()
		(q3) edge node [pos=0.5, sloped, below]{$a:0.1$} (q2)  
		edge [loop right,sloped,above] node {$a:0.9$} ()
		
		; 
		
		\node[text width=3cm] at (-1,-0.5) 
		{\{$i=0,!bad$\}};
		1	\node[text width=3cm] at (6,-0.5) 
		{\{$i<N-1,!bad$\}};
		\node[text width=3cm] at (4,-5) 
		{\{$i=N-1,!bad$\}};
		\node[text width=3cm] at (-2,-5) 
		{\{$i=N-1,bad$\}};	
		\end{tikzpicture} 
		
		\caption{The second abstract quotient automaton $\mathcal{M}^\#_2$.}
		\label{fig:aM2}
	\end{figure}
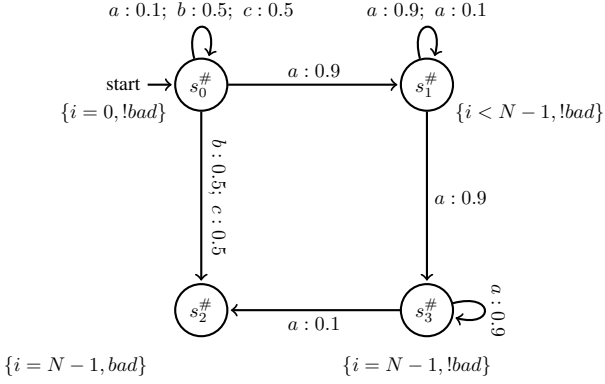

	\begin{figure}
		\centering
		\subfigure[$\mathcal{K}_1$]{		
			\begin{tikzpicture}[shorten >=1pt,node distance=4cm,on grid,auto, bend angle=20, thick,scale=0.75, every node/.style={transform shape}] 
			\node[state,initial,accepting] (s0)   {$s_0$}; 
			\node[state,accepting] (s1) [right= of s0] {$s_1$}; 		
			\path[->]

			(s0) edge node {$s_0a,s_0b$} (s1)

			(s1)
			edge [loop above, pos = 0.5,sloped, above] node {$s_*a,s_*b,s_*c$} ()
			
			; 
			\end{tikzpicture} 
		}
		\subfigure[$\mathcal{K}_2$]{		
			\begin{tikzpicture}[shorten >=1pt,node distance=4cm,on grid,auto, bend angle=20, thick,scale=0.75, every node/.style={transform shape}] 
			\node[state,initial,accepting] (s0)   {$s_0$}; 
			\node[state,accepting] (s1) [right= of s0] {$s_1$}; 		
			\path[->]

			(s0) edge node {$s^\#_0a,s^\#_0c$} (s1)

			(s1)
			edge [loop above, pos = 0.5,sloped, above] node {$s^{\#}_*a,s^\#_*b,s^\#_*c$} ()

			; 
			\end{tikzpicture} 
		}
		\caption{The resulting supervisors $\mathcal{K}_1$ and $\mathcal{K}_2$.}
		\label{fig:asupervisor2}
	\end{figure}
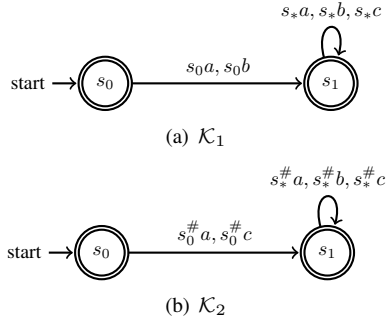
	
	Given the specification, our initial supervisors for both $\mathcal{M}_1$ and $\mathcal{M}_2$ enables all actions for any states. Then we have the initial abstract quotient automaton for $\mathcal{M}_2$ shown in Fig. \ref{fig:aM1}. This quotient automaton is a probabilistic automaton \cite{segala1995probabilistic} and can be seen as a special type of MDP where there may be many probability distributions defined for one action on a state. But the model checking problem on this quotient automaton is equivalent to MDP model checking \cite{hermanns2008probabilistic}.  In the first iteration, the model checking returns the counterexample path $(s_0,s^\#_0)c(s_2,s^\#_2)$. Since $c$ is the active action in $\mathcal{M}_1$. We select the first system to eliminate the projected path $s_0cs_2$. In the second iteration, the model checking on the abstract system first returns counterexample path as $s^\#_0as^\#_1as^\#_2$ with probability 0.81. Following algorithm in \cite{hermanns2008probabilistic}, we know this counterexample is introduced by the coarse abstraction and the refinement algorithm delivers new abstract system shown in Fig. \ref{fig:aM2}. Then the model checking on the new abstract system returns counterexample path $(s_0s^\#_0)b(s_2,s^\#_2)$ and this counterexample is a real counterexample. This time we refine the supervisor for $\mathcal{M}_2$ since $b$ is its active action. After the second iteration, no more counterexample is generated and our resulting supervisors $\mathcal{K}_1$ and $\mathcal{K}_2$ are as shown in Fig. \ref{fig:asupervisor2} where $s_*$ and $s_*^\#$ represents any $s$ and $s^\#$.

	\section{Conclusion}
	\label{sec:conclusion}
	In this paper, we proposed an automated framework of learning based  permissive supervisor synthesis for multi-agent systems. Assume-guarantee reasoning based model checking technique was applied to avoid the state space explosion problem. We also proposed to partition the action sets such that we know which subsystem is at fault when the specification is violated. It is guaranteed that we can get the correct supervisors in finite steps. For future research, we are interested in applying the same framework in partially known MDP models and more general class of specifications.

	\bibliographystyle{IEEEtran}
	\bibliography{ref1,ref}
\end{document}